\documentclass[12pt]{amsart}
\usepackage{graphicx} 
\usepackage[a4paper, margin=3cm]{geometry}
\usepackage{xcolor}
\usepackage{soul, url}
\usepackage{amsmath, amsfonts, amssymb, amsthm}
\usepackage{hyperref, comment}

\numberwithin{equation}{section}

\DeclareMathOperator{\wt}{wt}
\DeclareMathOperator{\maxwt}{maxwt}

\newtheorem{theorem}{Theorem}[section]
\newtheorem{proposition}[theorem]{Proposition}                                                      
\newtheorem{lemma}[theorem]{Lemma}
\newtheorem{corollary}[theorem]{Corollary}

\theoremstyle{definition}
\newtheorem{definition}[theorem]{Definition}
\newtheorem{remark}[theorem]{Remark}
\newtheorem{example}[theorem]{Example}

\newcommand{\rmv}[1]{{}}

\newcommand{\FF}{\mathbb{F}}

\title{Code distances: a new family of invariants of linear codes}

\author{Eduardo Camps-Moreno}
\address[Eduardo Camps-Moreno]{Department of Mathematics\\ Virginia Tech\\ 
Blacksburg, VA, USA}
\email{eduardoc@vt.edu}

\author{Elisa Gorla}
\address[Elisa Gorla]{Department of Mathematics\\ University of Neuchatel\\ Neuchatel, Switzerland}
\email{elisa.gorla@unine.ch}

\author{Hiram H. L\'opez}
\address[Hiram H. L\'opez]{Department of Mathematics\\ Virginia Tech\\ 
Blacksburg, VA, USA}
\email{hhlopez@vt.edu}

\thanks{Hiram H. L\'opez was partially supported by the NSF Grants DMS-2401558 and DMS-2502705. The authors thank Virginia Tech's Steger Center for International Scholarship, where some of the work contained in this paper was done.}

\subjclass[2020]{94B05, 94B75}

\keywords{Invariants of codes, minimum distance, generalized weights, covering radius, code distances.}

\begin{document}
\begin{abstract}
In this paper, we introduce code distances, a new family of invariants for linear codes. We establish some properties and prove bounds on the code distances, and show that they are not invariants of the matroid (for a linear block code) or $q$-polymatroid (for a rank-metric code) associated to the code. By means of examples, we show that the code distances allow us to distinguish some inequivalent MDS or MRD codes with the same parameters. We also show that no duality holds, i.e., the sequence of code distances of a code does not determine the sequence of code distances of its dual. Further, we define a greedy and an asymptotic version of code distances. Finally, we relate these invariants to other invariants of linear codes, such as the maximality degree, the covering radius,  and the partial distances of polar codes.
\end{abstract}

\maketitle

\section{Introduction}

Let $V$ be a vector space of finite dimension over a finite field $\FF_q$ and $d$ a metric on $V$. A linear code is a linear subspace of $V$. Some of the classical examples are when $V = \FF_q^n$ and $d$ is the Hamming metric (block code), $V=\mathbb{F}_q^{m\times n}$ and $d$ is the rank metric (rank-metric code), and when $V=\mathbb{F}_q^{m_1\times n_1}\times\ldots\times\FF_q^{m_\ell\times n_\ell}$ and $d$ is the sum-rank metric (sum-rank metric code). 

A fundamental problem in coding theory asks to distinguish codes that are not equivalent, i.e., codes for which there exists no linear isometry of the ambient space that sends one to the other.
This may sometimes be done by means of invariants, since codes with different invariants cannot be equivalent. Classical code invariants include the generalized weights and the weight distribution of a code. These are both invariants of the associated matroid (for a linear block code) or $q$-polymatroid (for a rank-metric code). Moreover,  duality holds for these invariants, meaning that the generalized weights of a code determine those of its dual and, similarly, the weight distribution of a code determines that of its dual. Other invariants, such as the covering radius of a code, are not invariants of the associated matroid or $q$-polymatroid.

In this paper, we define a new family of invariants of linear codes, which we call code distances, as they depend on the minimum distances of the subcodes and supercodes of the code in question. More precisely, if $C\subseteq V$ is a $k$-dimensional code, then for $1\leq i\leq k$, the $i$-th subcode distance of $C$ is
$$\alpha_i(C)=\max\{d_{\min}(D)\ :\ D\subseteq C,\ \dim D=i\}.$$
For $k\leq i\leq n$, the $i$-th supercode distance of $C$ is
$$\alpha_i(C)=\max\{d_{\min}(D)\ :\ D\supseteq C,\ \dim D=i\}.$$
These are genuinely new code invariants, as we show among other results that they are not invariants of the associated matroid or $q$-polymatroid. As such, they may allow us to distinguish inequivalent codes and indeed we show by means of examples that they may allow us to distinguish inequivalent MDS or MRD codes. 

The paper is structured as follows. Section~\ref{prelimiaries} contains some preliminary definitions and results. In Section~\ref{main} we define code distances and propose several examples, that illustrate the main ideas and properties of this new family of invariants. We show that the code distances form a non-increasing integer sequence and that some of them satisfy monotonicity properties (see Proposition~\ref{alphainequality} and Proposition~\ref{alphas of subcodes}, respectively). In Proposition~\ref{alpha for MDS}, Proposition~\ref{alpha for MRD}, and Proposition~\ref{alphaMSDR}, we use the Singleton Bound to derive upper bounds on the code distances. Theorem~\ref{thm:notinvariant} and Theorem~\ref{thm:polymatroid} show that code distances are not invariants of the matroid or $q$-polymatroid of the code. In Section~\ref{greedy}, we propose a greedy version of code distances and in Theorem~\ref{thm:polar} we compare it to the partial distances of polar codes. In Section~\ref{comparison} we compare the code distances to other code invariants. In particular, Proposition~\ref{prop:mu} gives a relation between one of the code distances and the maximality degree. Theorem~\ref{thm:radii} relates some of the greedy code distances to the covering radius of appropriate subcodes and some of the code distances to generalized radii, another new sequence of invariants that we introduce in this paper. Proposition~\ref{alpha strong} allows us to read some of the code distances from the parity check matrix of the code. Finally, in Section~\ref{asymtotic} we consider the code distances of the extension of a code to every field containing the field over which it is defined. This allows us to define an asymptotic version of code distances, which turn out to be related to generalized covering radii, see Theorem~\ref{thm:genradii}.

\section{Preliminaries}\label{prelimiaries}

In this section, we introduce the concepts and notation we will use for the rest of the manuscript.

Let $V$ an $N$-dimensional vector space over a finite field $\mathbb{F}_q$ and $$d : V\times V\rightarrow\mathbb{R}$$ be a metric over $V$. We are mostly interested in the following situations:
\begin{itemize}
    \item $V=\FF_q^n$ and $d$ is the Hamming metric,
    \item $V=\mathbb{F}_q^{m\times n}$ and $d$ is the rank metric, and
    \item $V=\mathbb{F}_q^{m_1\times n_1}\times\ldots\times\FF_q^{m_\ell\times n_\ell}$ and $d$ is the sum-rank metric.
\end{itemize}

Throughout the paper, $C$ is a {\bf linear code}, i.e. a linear subspace of $V$. When $V = \FF_q^n$ and $d$ is the Hamming metric, $C$ is known as a {\bf block code}. If $V=\mathbb{F}_q^{m\times n}$ and $d$ is the rank metric, $C$ is a {\bf rank-metric code}. When $V=\mathbb{F}_q^{m_1\times n_1}\times\ldots\times\FF_q^{m_\ell\times n_\ell}$ and $d$ is the sum-rank metric, $C$ is known as a {\bf sum rank-metric code}.

Let $k$ be the {\bf dimension} of $C$ as an $\FF_q$-linear space. For any $v\in V$, we let $$wt(v)=d(v,0)$$ be the {\bf weight} of $v$. For any nontrivial code $0\neq C\subseteq V$, we denote its {\bf minimum distance} by
$$d_{\min}(C)=\min\{d(w,v)\ :\ w,v\in C,\ w\neq v\}=\min\{\wt(v)\ :\ v\in C\setminus 0\}$$
and its {\bf maximum weight} by $$\maxwt(C)=\max\{\wt(v)\ :\ v\in C\}=\max\{d(w,v)\ :\ w,v\in C\}.$$

The minimum distance is one of the main parameters of a code. It is well-known that it measures the error-correction capability of the code, see~\cite{HP03}. The Singleton Bound is one of the most classical results capturing a relation between some of the the main parameters of a code. For a code $C \subseteq V$ of dimension $k$, the bound states the following.
\begin{itemize}
\item If $V=\FF_q^n$ and $d$ is the Hamming distance, then $$d_{\min}(C) \leq n-k+1.$$ A code that meets the bound is called a {\bf maximum-distance separable (MDS) code}.
\item If $V=\mathbb{F}_q^{m\times n}$, $n\leq m$, and $d$ is the rank metric, then
$$d_{\min}(C) \leq n - \lceil k/m \rceil + 1.$$
If $m$ divides $k$, a code that meets the bound is called a {\bf maximum rank distance (MRD) code}. If $m$ does not divide $k$, a code achieving the bound is called a {\bf quasi-maximum rank distance (QMRD) code}.
\item If $V=\mathbb{F}_q^{m_1\times n_1}\times\ldots\times\FF_q^{m_\ell\times n_\ell}$, $n_i\leq m_i$ for all $i$, and $d$ is the sum-rank metric, let $j$ and $\delta$ be the unique integers satisfying $d_{\min}(C)-1 = \sum^{j-1}_{i=1} n_i + \delta$ and $0 \leq \delta \leq n_j-1$. Then
\[k \leq \sum_{i=j}^\ell m_in_i - m_j\delta. \]
A code that achieves the bound is called a {\bf maximum sum-rank distance (MSRD) code}.
\end{itemize}

The ring of univariate polynomials over $\FF_q$ is denoted by $\FF_q[x]$. Given a positive integer $k$, $\FF_{q}[x]_{<k}$ denotes the set of polynomials of degree smaller than $k$. A subset $A=\{a_1,\ldots, a_n\}\subseteq\FF_q$ defines an {\bf evaluation map}
\[\begin{array}{lccc}
&\FF_{q}[x]_{<k} & \to & \FF_{q}^n \\
& f & \mapsto & f(A)=\left(f(a_1),\ldots, f(a_n)\right).
\end{array}\]

\begin{definition}
Let $A=\{a_1,\ldots, a_n\}\subseteq\FF_q$. The {\bf Reed-Solomon (RS) code} is the image of the evaluation map
\[ RS(A,k) = \left\{ f(A) : f \in \FF_{q}[x]_{<k} \right\}.\]
A Reed-Solomon $RS(A,k)$ has minimum distance $n-k+1$, i.e., it is MDS.
\end{definition}
\begin{definition}
Let $m, n, k$ be positive integers such that $m \geq n \geq k$, and let $\{a_1, \ldots, a_n\}\subseteq\FF_{q^m}$ be a set of linearly independent elements over $\FF_{q}$. The $\FF_{q^m}$-linear code generated by the matrix 
\[G = \begin{pmatrix}
a_1 & a_2 & \ldots & a_n\\
a_1^q & a_2^q & \ldots & a_n^q\\
\vdots & \vdots & \ddots & \vdots\\
a_1^{q^{k-1}} & a_2^{q^{k-1}} & \ldots & a_n^{q^{k-1}}
\end{pmatrix}\]
is called a {\bf Gabidulin code} of
length $n$ and dimension $k$.
\end{definition}
It is well-known, and not hard to show, that a Gabidulin code is MRD.

Let $V=\FF_q^n$ and let $d$ be the Hamming distance. Let $C \subseteq V$ be a code of dimension $k$. A {\bf generator matrix} of $C$ is a $k \times n$ matrix such that $C$ is the rowspace of $G$. The {\bf dual} of $C$ is 
\[ C^\perp = \{ c \in \FF_q^{n} : c \cdot c^\prime = 0, \textrm{ for all } c^\prime \in C \},\]
where $c \cdot c^\prime$ is the Euclidean inner product. A {\bf parity check} matrix of $C$ is an $(n-k)\times n$ matrix $H$, whose row space generates $C^{\perp}$. If $C = \FF_q^n$, we let the parity check matrix be $H=\left( 0\; \cdots\; 0 \right) \in \FF_q^n$. A classical result in coding theory relates the minimum distance of $C$ with the minimum number of linearly dependent columns of $H$. Specifically, we have the following result.

\begin{proposition}\cite[Corollary 1.4.14]{HP03}\label{25.02.26}
Let $V=\FF_q^n$ and let $d$ be the Hamming distance. Let $C \subseteq V$ be a linear code and $H$ its parity check matrix. The minimum distance of $C$ is the minimum number of linearly dependent columns of $H$.
\end{proposition}

Recall that elements in $V$ are called linearly dependent (l.d.) if there are coefficients in $\FF_q$, not all zero, such that the corresponding linear combination is zero. If we impose the condition that all the coefficients are nonzero, we have the following concept.

\begin{definition}\label{25.02.29}
We say that the elements $v_1\ldots,v_n \in V$ are {\bf strongly linearly dependent} (s.l.d.) if there exist nonzero elements $\lambda_1,\ldots,\lambda_n$ in $\FF_q$ such that
\[\lambda_1 v_1 + \cdots + \lambda_n v_n = 0.\]
\end{definition}
It is clear that s.l.d. elements are l.d. The next example shows that the converse is not necessarily true.

\begin{example}
The elements $v_1=(1,1,0,0)$, $v_2=(0,1,1,0)$, $v_3=(1,0,1,0)$ and $v_4=(0,0,0,1)$ of $\FF_2^4$ are linearly dependent because $1v_1+1v_2+1v_3+0v_4=0$, but they are not strongly linearly dependent because $1v_1+1v_2+1v_3+1v_4 \neq 0$.
\end{example}

One can easily relate the maximum weight of a code with the maximum number of strongly linearly dependent columns of its parity check matrix. The proof follows directly from the definition.

\begin{proposition}\label{25.02.27}
Let $V=\FF_q^n$ and let $d$ be the Hamming distance. Let $C \subseteq V$ be a linear code and $H$ be its parity check matrix. The maximum weight of $C$ is the maximum number of strongly linearly dependent columns of $H$.
\end{proposition}

We observe that, for a given set of vectors, the minimum cardinality of a linearly dependent set of elements and the minimum cardinality of a strongly linearly dependent set of elements coincide.

\begin{lemma}\label{25.02.28}
Let $v_1,\ldots,v_t$ be elements in $\FF_q^n$. The minimum cardinality of a subset of linearly dependent elements coincides with the minimum cardinality of a subset of strongly linearly dependent elements.
\end{lemma}

\begin{proof}
By Definition~\ref{25.02.29}, s.l.d. implies l.d. We now check the converse. Let $m$ be the minimum number of l.d. elements. Without loss of generality, assume that $v_1,\ldots,v_m$ are l.d. Then there exist $\lambda_1,\ldots,\lambda_m$ in $\FF_q$ such that
\[\lambda_1v_1 + \cdots + \lambda_mv_m = 0.\]
If one of the $\lambda_i$'s is $0$, then the minimum number of l.d. elements is smaller than $m$. Hence $v_1,\ldots,v_m$ are s.l.d.
\end{proof}

The next result follows by combining Proposition~\ref{25.02.26}, Proposition~\ref{25.02.27}, and Lemma~\ref{25.02.28}.

\begin{corollary}\label{25.03.01}
Let $V=\FF_q^n$ and let $d$ be the Hamming distance. Let $C \subseteq V$ be a linear code with parity check matrix $H$. The maximum and minimum weights of $C$ are, respectively, the maximum and minimum number of strongly linearly dependent columns of $H$.
\end{corollary}

Finally, for the convenience of the reader we recall the definition of matroid and $q$-matroid. Matroids were introduced by Whitney in~\cite{W35}. The book~\cite{O11} is a standard reference on the subject.

\begin{definition}
A \textbf{matroid} is a pair $(\{1,\ldots,n\},\rho)$ where $n \ge 1$ and $\rho$ is a function from the set of all subsets of $\{1,\ldots,n\}$ to $\mathbb{Z}$ such that, for all $A,B\subseteq \{1,\ldots,n\}$:
\begin{itemize} \setlength\itemsep{0em}
\item[(P1)] $0\leq \rho(A)\leq |A|$,
\item[(P2)] if $A\subseteq B$, then $\rho(A)\leq \rho(B)$,
\item[(P3)] $\rho(A+B)+\rho(A\cap B)\leq \rho(A)+\rho(B)$.
\end{itemize}
\end{definition}

$q$-polymatroids were introduced independently in~\cite{GJLR20} and~\cite{S19}. 

\begin{definition}[{\cite[Definition 4.1]{GJLR20}}]
A \textbf{$q$-polymatroid} is a pair $(\FF_q^n,\rho)$ where $n \ge 1$ and $\rho$ is a function from the set of all subspaces of $\FF_q^n$ to $\mathbb{R}$ such that, for all $A,B\subseteq \FF_q^n$:
\begin{itemize} \setlength\itemsep{0em}
\item[(P1)] $0\leq \rho(A)\leq\dim(A)$,
\item[(P2)] if $A\subseteq B$, then $\rho(A)\leq \rho(B)$,
\item[(P3)] $\rho(A+B)+\rho(A\cap B)\leq \rho(A)+\rho(B)$.
\end{itemize}
\end{definition}


\section{Main properties}\label{main}

In this section, we introduce code distances, a new family of invariants for linear codes. We will show that, in the case of linear block codes or rank-metric codes, the code distances are not invariants of the matroid or $q$-polymatroid associated with the code. In particular, the code distances sometimes allow us to distinguish codes that cannot be distinguished by looking at (any invariant of) the associated matroid or $q$-polymatroid. For example, code distances will enable us to distinguish some inequivalent MDS and MRD codes.

By considering the sequence of minimum distances of subcodes, we obtain a natural family of invariants of a code, which contains both its minimum distance and maximum weight.

\begin{definition}[Subcode distances]
Let $C\subseteq V$ be a $k$-dimensional code. For $1\leq i\leq k$, we define the {\bf $i$-th subcode distance} of $C$ as
$$\alpha_i(C)=\max\{d_{\min}(D)\ :\ D\subseteq C,\ \dim D=i\}.$$
\end{definition}

We can define further invariants by considering supercodes of $C$.

\begin{definition}[Supercode distances]
Let $C\subseteq V$ be a $k$-dimensional code. For $k\leq i\leq n$, we define the {\bf $i$-th supercode distance} of $C$ as
$$\alpha_i(C)=\max\{d_{\min}(D)\ :\ D\supseteq C,\ \dim D=i\}.$$
We refer to the subcode and supercode distances as the {\bf code distances}.
\end{definition}

Recall that two linear codes $C, D \subseteq V$ are {\bf equivalent} if there is a distance-preserving, $\FF_q$-linear isomorphism of $V$ with itself that maps $C$ to $D$. Observe that code distances are invariant under equivalence.

\begin{example}\label{25.01.01}
Let $V=\mathbb{F}_2^3$ and let $d$ be the Hamming distance. If $C=\mathbb{F}_2^3$, then
$$\begin{array}{ccc}
\alpha_1(C)=&d_{\min}(\langle (1,1,1)\rangle)&=3,\\
\alpha_2(C)=&d_{\min}(\langle (1,1,0),\ (0,1,1)\rangle)&=2,\\
\alpha_3(C)=&d_{\min}(C)&=1.
\end{array}$$
\end{example}

\begin{example}\label{25.01.02}
Let $V=\mathbb{F}_2^4$ and let $d$ be the Hamming distance. If $C=\mathbb{F}_2^4$, then
$$\begin{array}{ccc}
\alpha_1(C)=&d_{\min}(\langle (1,1,1,1)\rangle)&=4,\\
\alpha_2(C)=&d_{\min}(\langle (1,1,0,0),\ (0,0,1,1)\rangle)&=2,\\
\alpha_3(C)=&d_{\min}(\langle (1,1,0,0),\ (0,1,1,0),\ (0,0,1,1)\rangle)&=2,\\
\alpha_4(C)=&d_{\min}(C)&=1.
\end{array}$$
\end{example}

We now establish some basic properties of code distances.

\medskip

It is immediate from the definition that
\[
\alpha_1(C)=\maxwt(C) \;\mbox{ and }\; \alpha_k(C)=d_{\min}(C).
\]

\begin{proposition}\label{alphainequality}
Let $C \subseteq V$ be a $k$-dimensional code. Then,
$$\maxwt(C)=\alpha_1(C)\geq \alpha_2(C)\geq\cdots\geq\alpha_k(C)=d_{\min}(C)\geq \cdots\geq\alpha_n(C)=d_{\min}(V).$$
\end{proposition}
\begin{proof}
For $1\leq i\leq k-1$, let $D\subseteq C$ have $\dim(D)=i+1$ and $d_{\min}(D)=\alpha_{i+1}(C)$. Let $E\subseteq D$ be such that $\dim(E)=i$. Then,
$$\alpha_{i+1}(C)=d_{\min}(D)\leq d_{\min}(E)\leq \alpha_i(C).$$ For $k\leq i\leq n$, let $D\supseteq C$ have $\dim(D)=i+1$ and $d_{\min}(D)=\alpha_{i+1}(C)$. Let $C\subseteq E\subseteq D$ be such that $\dim(E)=i$. Then,
$$\alpha_{i+1}(C)=d_{\min}(D)\leq d_{\min}(E)\leq \alpha_i(C),$$
which completes the proof.
\end{proof}

The following examples show that the sequence of code distances is not, in general, strictly decreasing. For instance, the sequence can contain a constant string of arbitrary length. This behavior is in contrast with the generalized weights, which are strictly increasing in the Hamming metric, and which cannot contain constant strings of arbitrary length in the rank and sum-rank metrics.

\begin{example}\label{25.03.02}
The $k$-dimensional $q$-ary simplex code $C\subseteq\FF_q^n$, $n=q^{k-1}+\ldots+q+1$ is a linear code all of whose nonzero codewords have weight $q^{k-1}$. Hence, it has
$$\alpha_1(C)=\cdots=\alpha_k(C)=q^{k-1}.$$
For example, the $k$-dimensional binary simplex code $C\subseteq\FF_2^{2^k-1}$ has
$$\alpha_1(C)=\cdots=\alpha_k(C)=2^{k-1}.$$
\end{example}

Linear constant weight codes are also called equidistant codes. In the Hamming metric, they were characterized in~\cite{B84}, see also~\cite{KK23}.

\begin{example}\label{ex:constant}
Every $k$-dimensional constant weight code $C\subseteq\FF_q^n$ has $d_{\min}(C)=dq^{k-1}$ and is equivalent to the $d$-fold repetition of the $q$-ary simplex code, possibly extended by zeroes. This implies in particular that $n\geq d(q^{k-1}+\ldots+q+1)$. The code distances of such a code are
$$\alpha_1(C)=\ldots=\alpha_k(C)=dq^{k-1}.$$
\end{example}

One can easily compare some of the code distances of a code with those of a subcode.

\begin{proposition}\label{alphas of subcodes}
Let $C\subseteq D\subseteq V$ be linear codes. Then $$\alpha_i(C)\leq\alpha_i(D)\; \mbox{ for } 1\leq i\leq \dim(C)$$ and $$\alpha_i(C)\geq\alpha_i(D)\; \mbox{ for } \dim(D)\leq i\leq n.$$  
\end{proposition}

\begin{proof}
For $1\leq i\leq\dim(C)$, the thesis follows from observing that every subcode of $C$ is a subcode of $D$. For $\dim(D)\leq i\leq n$, the result follows from observing that every supercode of $D$ is a supercode of $C$. 
\end{proof}

The next example shows that Proposition~\ref{alphas of subcodes} cannot be extended. In other words, if $i$ is not in the appropriate range, no inequality between the $i$-th code distances of two nested codes may hold.

\begin{example}\label{25.03.04}
Let
\[C = \langle(1,1,1,0,0)\rangle\subseteq D=\langle(1,1,1,0,0),(0,0,1,1,0),(0,0,0,1,1)\rangle \subseteq\FF_2^5.\]
One can check that $$\alpha_2(C) = 3 > 2 = \alpha_2(D).$$
Let \[C = \langle(1,1,0,0,0)\rangle \subseteq D=\langle(1,1,1,0,0),(0,0,1,1,1),(1,1,0,0,0)\rangle \subseteq\FF_2^5.\]
An easy computation shows that $$\alpha_2(C) = 2 < 3 = \alpha_2(D).$$
\end{example}

Since computing the minimum distance of a code is NP-hard by~\cite{V97}, computing the sequence of code distances is also NP-hard. Notice that, since~\cite{V97} applies to linear block codes endowed with the Hamming weight, then it also applies to sum-rank metric codes, of which linear block codes are a special case. Moreover, since one can represent linear block codes as rank-metric codes by mapping the entries of a vector to the entries on the diagonal, the result also applies to rank-metric codes.

\subsection{Hamming distance} 

In this section we focus on the case when $V=\FF_q^n$ and $d$ is the Hamming distance. Let $C \subseteq V$ be a linear code. The Singleton Bound applied to the subcodes or supercodes of $C$ provides a simple bound for the code distances.

\begin{proposition}\label{alpha for MDS}
Let $C \subseteq V$ be a code of dimension $k$. Then $$\alpha_i(C)\leq n-i+1$$ for $1\leq i\leq n$. Moreover
\begin{itemize}
\item For $1\leq i \leq k-1$, then $\alpha_i(C)=n-i+1$ if and only if $C$ contains an $i$-dimensional MDS code.
\item $\alpha_k(C)=n-k+1$ if and only if $C$ is MDS.
\item For $k+1 \leq i \leq n$, then $\alpha_i(C)=n-i+1$ if and only if $C$ is contained in an $i$-dimensional MDS code.
\end{itemize}
\end{proposition}

When $q\geq n$, the Reed-Solomon codes provide examples of $i$-dimensional MDS codes in $\FF_q^n$ for any $1\leq i\leq n$. Thus, we obtain the following result.

\begin{corollary}\label{alpha RS code}
Assume $q \geq n$. If $C \subseteq V$ is a Reed-Solomon code, then
\[\alpha_i(C)=n-i+1 \quad \text{ for } \quad 1\leq i\leq n.\]
 In particular,
 \[\alpha_i(\FF_q^n)=n-i+1 \quad \text{ for } \quad 1\leq i\leq n.\]
\end{corollary}

A consequence of the last result is that an optimal anticode of dimension $k$ in $\FF_q^n$ has $\alpha_i=k-i+1$, provided that $q\gg 0$. For $q=2$, we have $\alpha_1(\FF_2^n)=n$ and $\alpha_2(\FF_2^n)=\lfloor 2n/3 \rfloor$.

\begin{remark}\label{25.03.05}
Notice that the code distances of a $k$-dimensional Reed-Solomon code do not depend on $k$. In particular, the code distances of a code do not determine its dimension or its minimum distance. The minimum distance is however determined, if the dimension is known.
\end{remark}

The next example is a code $C \subseteq V$ such that $\alpha_1(C)=n$ and $\alpha_k(C)=n-k+1$, which means that they both achieve the bound of Proposition~\ref{alpha for MDS}. However, $\alpha_i(C)$ does not achieve the bound for $1 < i < k$. In other words, there are MDS codes of maximum weight equal to the length, whose code distances are not all maximal. 

\begin{example}\label{ex:ternary422}
Let $V=\mathbb{F}_3^4$, $d$ the Hamming distance, and  
\[C=\langle (1,1,1,0),\ (0,1,2,0),\ (0,0,1,1)\rangle.\]
Then $C$ is a doubly extended Reed-Solomon code of dimension 3, whose dual is $\langle (1,1,1,2)\rangle$. One can easily check that $\alpha_1(C)=4$ and $\alpha_3(C)=2$. For any $D\subseteq C$ of dimension 2, the dual $D^\perp$ has a generator matrix of the form
$\begin{pmatrix}
1&1&1&2\\
a_1&a_2&a_3&a_4
\end{pmatrix}.$
If the elements $a_1,a_2,a_3$ are different, then the fourth column and one of the first three columns are linearly dependent. If two of the elements $a_1,a_2,a_3$ are equal, then two of the first three columns are linearly dependent. In either case, $d_{\min}(D)=2$ by Proposition~\ref{25.02.26}. Hence $\alpha_2(C)=2$.
Summarizing, we obtain
$$\begin{array}{ccc}
\alpha_1(C)=&d_{\min}(\langle (1,1,2,1)\rangle)&=4,\\
\alpha_2(C)=&d_{\min}(\langle (1,1,1,0),\ (0,0,1,1)\rangle)&=2,\\
\alpha_3(C)=&d_{\min}(C)&=2.
\end{array}$$
\end{example}

\begin{remark}
We already observed that the sequence of code distances of a code $C$ may contain a constant subsequence. However, $C$ may or may not contain a constant weight subcode. For example, the even weight code in $\FF_2^4$ whose parity check matrix is $(1,1,1,1)$ and the code $C\subseteq\mathbb{F}_3^4$ of Example~\ref{ex:ternary422} have the same code distances:
\[\alpha_1(C)=4, \qquad \alpha_2(C)=2, \qquad \alpha_3(C)=2, \qquad \text{ and } \qquad \alpha_4(C)=1.\]
Nevertheless, the even weight code contains the two-dimensional constant weight code $\langle (1,1,0,0), (0,1,1,0)\rangle$, while $C$ does not contain any two-dimensional constant weight code of weight two. In fact, no such code exists for $q=3$ (see Example~\ref{ex:constant}).
Thus, the existence of a constant subsequence in the code distances is not necessarily connected to the existence of a constant weight subcode.
\end{remark}

\begin{remark}\label{rmk:subcode}
As there is no code in $\mathbb{F}_3^4$ whose sequence of subcode distances is $2,2$, it is not true in general that every subsequence of the sequence of subcode distances of a code $C\subseteq V$ is itself the sequence of subcode distances of some code $D\subseteq V$.
An even simpler example is $\FF_2^3$, whose sequence of subcode distances is $3,2,1$, but which does not contain any code whose sequence of subcode distances is $3,2$. This is in contrast with the case of generalized weights of a linear block code, rank-metric or sum-rank metric code, for which such a result holds under suitable assumptions, see~\cite[Theorem~3.1, Theorem~4.10, and Theorem~4.13]{GLMS23}.
\end{remark}

Next we give an example of MDS codes with the same length, dimension, and minimum distance, but different code distances. This shows that code distances may allow us to distinguish non-equivalent MDS codes.

\begin{example}[Reed-Solomon codes and twisted Reed-Solomon codes]\label{ex:RSvsERS}
Let $C=RS(\FF_9,4)$ be a Reed-Solomon code and let $D = \left\{ f(\FF_9) : f \in \langle 1,x,x^2-ax^6,x^3 \rangle \right\},$ where $a$ is a primitive element of $\mathbb{F}_9$. Both are MDS codes over $\FF_9$ of length 9 and dimension 4. However,
$$\alpha_1(C)=9,\ \alpha_2(C)=8,\ \alpha_3(C)=7,\ \alpha_4(C)=6,$$
$$\alpha_1(D)=9,\ \alpha_2(D)=8,\ \alpha_3(D)=6,\ \alpha_4(D)=6.$$
This proves that $C$ and $D$ are not equivalent. 
\end{example}

\begin{theorem}\label{thm:notinvariant}
Let $V=\mathbb{F}_q^n$ and $d$ be the Hamming metric. The sequence of code distances is not an invariant of the matroid associated with a code.
\end{theorem}

\begin{proof}
Since the matroid associated with an MDS code is determined by the code parameters, Example~\ref{ex:RSvsERS} shows that code distances are not invariants of the matroid associated with the code. 
\end{proof}

The sequence of code distances is an invariant of a linear block code which is not an invariant of the associated matroid. Other such invariants are the covering radius, the Netwon radius, and the dimension of the hull. The maximum weight of a code is also not an invariant of the matroid associated with the code. For a given $q$, however, it is determined by the matroid, as the weight enumerator is. For $q=2$, the circuits of the matroid generate the code to which it is associated. In particular, every invariant of a binary code is an invariant of the associated matroid, as the code and the matroid are equivalent information in this case.

\medskip

It is well-known that the generalized weights of a linear block code determine those of its dual. A similar result holds for rank-metric codes. This is not the case, however, for the generalized weights of sum-rank metric codes, see~\cite[Example V.10]{CGLLMS22}. The following example shows that the code distances of a linear block code do not determine the code distances of the dual.

\begin{example}\label{ex:duality}
Let $C_1,C_2\subseteq\FF_2^5$ be the binary codes with generator matrices 
$$G_1=\begin{pmatrix} 1&1&1&1&0\\
0&0&0&1&1
\end{pmatrix} \qquad \text{ and } \qquad
G_2=\begin{pmatrix}1&1&1&1&0\\
0&0&1&1&0\end{pmatrix},$$
respectively. It is easy to check that 
\[\alpha_1(C_i)=4, \quad \alpha_2(C_i)=2, \quad \alpha_3(C_i)=2, \quad \alpha_4(C_i)=2, \quad \text{ and } \quad \alpha_5(C_i)=1\] for $i=1,2$.
The dual codes $C_1^\perp$ and $C_2^\perp$ are generated, respectively, by
$$G_1^\perp=\begin{pmatrix} 1&1&0&0&0\\
0&1&1&0&0\\
0&0&1&1&1
\end{pmatrix} \qquad \text{ and } \qquad 
G_2^\perp=\begin{pmatrix}0&0&0&0&1\\
1&1&0&0&0\\
0&0&1&1&0\end{pmatrix}.$$
One can check that
\[\alpha_1(C_1^\perp)=5, \quad \alpha_2(C_1^\perp)=3, \quad \alpha_3(C_1^\perp)=2, \quad \alpha_4(C_1^\perp)=1, \quad \text{ and } \quad \alpha_5(C_1^\perp)=1\]
\[\alpha_1(C_2^\perp)=5, \quad \alpha_2(C_2^\perp)=2, \quad \alpha_3(C_2^\perp)=1, \quad \alpha_4(C_2^\perp)=1, \quad \text{ and } \quad \alpha_5(C_2^\perp)=1.\]
\end{example}

The next proposition easily follows from the definition.

\begin{proposition}
Let $D\subseteq\FF_q^{n-1}$ be a puncturing of $C\subseteq\FF_q^n$. Then,
$$\alpha_i(D)\geq\alpha_i(C)-1.$$
\end{proposition}

The next example shows that the code distances of $C$ do not determine the code distances of the puncturings of $C$. This is in contrast with what happens for the matroid associated with a code, which determines the matroids of all the puncturings of the code. The example also shows that the code distances do not detect whether the code is degenerate.

\begin{example}\label{ex:shortening}
Let $C_1$ and $C_2$ be the binary codes defined in Example~\ref{ex:duality}. The shortening of the last coordinate of $C_2$ yields a code of dimension $2$, while every shortening of $C_1$ has dimension $1$. In addition, puncturing $C_2$ in the last coordinate yields a code with minimum distance $2$ and maximum weight $4$. Puncturing $C_1$ in one of the first three coordinates yields a code with minimum distance $2$ and maximum weight $3$. Finally, puncturing $C_1$ in one of the last two coordinates yields a code with minimum distance $1$ and maximum weight $4$.
\end{example}

\begin{example}
Let $C\subseteq \mathbb{F}_2^n$ be the even weight code, i.e., the parity check matrix of $C$ is $H=(1,\ldots, 1) \in \FF_2^n$. Then
$$\begin{array}{rcl}
\alpha_1(C)=&\maxwt(C)&=2 \displaystyle \left\lfloor n/2\right\rfloor,\\
\alpha_2(C)=&&=2 \displaystyle \left\lfloor n/3 \right\rfloor,\\
\alpha_{n-3}(C)=& &=2 \text{ if } n\geq 5,\\
\alpha_{n-2}(C)=& &=2 \text{ if } n\geq 3,\\
\alpha_{n-1}(C)=&d_{\min}(C)&=2,\\
\alpha_n(C)=&d_{\min}(\FF_2^n)&=1.
\end{array}$$
The computations of $\alpha_1(C)$, $\alpha_{n-1}(C)$, and $\alpha_n(C)$ are straightforward from the definitions of subcode and supercode distances. If $n=2$, then these are all the code weights. Hence, from now on, suppose that $n\geq 3$.

For brevity, we write $\alpha_2$ instead of $\alpha_2(C)$. Any two-dimensional subcode of $C$ has a basis of two codewords of weights $\alpha_2$ and $w$ with $\alpha_2\leq w$ and whose supports intersect in $t$ components, where $\alpha_2+w-t\leq n$ and $\alpha_2+w-2t\geq\alpha_2$. It follows that $2\alpha_2+2w-2n\leq w$, from which $\alpha_2\leq w\leq 2n-2\alpha_2$. Since $\alpha_2$ is even, this shows that $\alpha_2\leq 2\left\lfloor\frac{n}{3}\right\rfloor$. Equality follows from checking that the subcode of $C$ generated by the two codewords of weight $\alpha_2$ supported on the first and the last $\alpha_2$ entries, respectively, has minimum distance $\alpha_2$.

The computation of $\alpha_{n-2}(C)$ follows from observing that, since $C$ contains the two-dimensional constant weight code $\Lambda=\langle (1,1,0,\ldots,0),(0,1,1,0,\ldots,0)\rangle$, every $(n-2)$-dimensional subcode $D$ of $C$ must intersect $\Lambda$ nontrivially, hence it must contain a codeword of weight $2$.

We now compute $\alpha_{n-3}(C)$, under the assumption that $n \geq 5$. Let $D$ be an $(n-3)$-dimensional subcode of $C$ and consider again the two-dimensional constant weight code $\Lambda=\langle (1,1,0,\ldots,0),(0,1,1,0,\ldots,0)\rangle$. If $D\cap\Lambda\neq 0$, then $d_{\min}(D)=2$. Else, let $u=(1,1,1,1,0,\ldots,0)$, $v=(0,1,1,1,1,0,\ldots,0)$, and observe that $u,v\in C=D+\Lambda$. If $u\not\in D$, then $u+x\in D$ for some $x\in \Lambda\setminus 0$ and $\wt(u+x)=2$. We can argue similarly for $v$ and deduce that if $v\not\in D$, then $D$ contains a codeword of weight $2$. However, if $u,v\in D$, then $u+v\in D$ has $\wt(u+v)=2$. In either case $d_{\min}(D)=2$.

\smallskip

Moreover, for $n=7$ we have $\alpha_1(C)=6$, $\alpha_2(C)=4$, $\alpha_3(C)=\alpha_4(C)=\alpha_5(C)=\alpha_6(C)=2$. Finally, for $n=8$ we have $\alpha_1(C)=6$, $\alpha_2(C)=\alpha_3(C)=4$, $\alpha_4(C)=\alpha_5(C)=\alpha_6(C)=\alpha_7(C)=2$.
\end{example}

\subsection{Rank and sum-rank metric}

In this section we focus on the case when $V=\mathbb{F}_q^{m\times n}$, $n\leq m$, and $d$ is the rank metric.
We start by observing that the Singleton Bound applied to the subcodes or supercodes provides a bound for the code distances. 

\begin{proposition}\label{alpha for MRD}
Let $V=\mathbb{F}_q^{m\times n}$, $n\leq m$, and let $d$ be the rank metric. Let $C \subseteq V$ be a code of dimension $k$. Then $$\alpha_i(C)\leq n-\left\lfloor\frac{i}{m}\right\rfloor+1$$ for $1\leq i\leq k$. Moreover
\begin{itemize}
\item If $1\leq i \leq k$, then $\alpha_i(C)= n-\left\lfloor\frac{i}{m}\right\rfloor+1$ if and only if $C$ contains an $i$-dimensional quasi-MRD code. If $m\mid i$, this is equivalent to $C$ containing an $i$-dimensional MRD code.
\item If $k \leq i \leq n$, then $\alpha_i(C)= n-\left\lfloor\frac{i}{m}\right\rfloor+1$ if and only if $C$ is contained in an $i$-dimensional quasi-MRD code. If $m\mid i$, this is equivalent to $C$ being contained in an $i$-dimensional MRD code.
\end{itemize}
\end{proposition}

\begin{corollary}\label{25.03.06}
Let $V=\mathbb{F}_q^{m\times n}$, $n\leq m$, and let $d$ be the rank metric. Let $C\subseteq\FF_q^{m\times n}$ be a Gabidulin code. Then
\[\alpha_i(C)=n-\left\lfloor\frac{i}{m}\right\rfloor+1 \quad \text{ for } \quad 1\leq i\leq mn.\]
\end{corollary}

\begin{proof}
This follows from Proposition~\ref{alpha for MRD}, because a Gabidulin code is an MRD code which fits into a maximal a chain of MRD codes, one for each dimension.
\end{proof}

\begin{remark}
Similarly to Remark~\ref{25.03.05}, Corollary~\ref{25.03.06} states that the code sequences do not determine the dimension or the minimum distance of a rank-metric code.
\end{remark}

A similar result holds for the sum-rank metric. The statement follows directly from the Singleton Bound as stated in~\cite[Corollary VI.5]{CGLLMS22}.

\begin{proposition}\label{alphaMSDR}
Let $V=\mathbb{F}_q^{m_1\times n_1}\times\ldots\times\mathbb{F}_q^{m_\ell\times n_\ell}$ and let $d$ be the sum-rank metric. Let $C \subseteq V$ be a code of dimension $k$. For $1\leq i\leq\sum_{h=1}^\ell m_i n_i$, let 
$$j=\min\{h\mid 1\leq h\leq\ell,\; i>\sum_{h=j+1}^{\ell} m_h n_h+m_j\}.$$
and $$\delta=\min\{h\mid 1\leq h\leq n_j-1,\; i>\sum_{h=j}^{\ell} m_h n_h-m_j\delta\}.$$
Then $$\alpha_i(C)\leq 1+\sum_{h=1}^{j-1} n_h+\delta.$$ 
Moreover
\begin{itemize}
\item If $1\leq i\leq k$, then $\alpha_i(C)=1+\sum_{h=1}^{j-1} n_h+\delta$ if and only if $C$ contains an $i$-dimensional MSRD code.
\item If $k \leq i\leq\sum_{h=1}^\ell m_i n_i$, then $\alpha_i(C)=1+\sum_{h=1}^{j-1} n_h+\delta$ if and only if $C$ is contained in an $i$-dimensional MSRD code.
\end{itemize}
\end{proposition}

The following is an example of two MRD codes with the same dimension and minimum distance, but different code distances. It follows from~\cite[Corollary~6.6]{GJLR20} that the $q$-polymatroids associated to $C_1$ and $C_2$ coincide. In fact, they both are the uniform $q$-matroid, see also~\cite[Example ~4.16]{JP16}.

\begin{example}\label{ex:unequivMRD}
Let $C_1\subseteq\FF_2^{4\times 4}$ be the code from~\cite[Example~7.2]{BR17}, i.e.,
$$C_1=\left\langle\begin{pmatrix}
   1 & 0 & 0 & 0 \\ 0 & 0 & 0 & 1 \\ 0 & 0 & 1 & 0 \\ 0 & 1 & 0 & 0
  \end{pmatrix}, \ \ \ \ \ 
\begin{pmatrix}
   0 & 1 & 0 & 0 \\ 0 & 0 & 1 & 1 \\ 0 & 0 & 0 & 1 \\ 1 & 1 & 0 & 0
  \end{pmatrix},  \ \ \ \ \ 
\begin{pmatrix}
   0 & 0 & 1 & 0 \\ 0 & 1 & 1 & 1 \\ 1 & 0 & 1 & 0 \\ 1 & 0 & 0 & 1
  \end{pmatrix},\ \ \ \ \ 
\begin{pmatrix}
   0 & 0 & 0 & 1 \\ 1 & 1 & 1 & 0 \\ 0 & 1 & 0 & 1 \\ 0 & 1 & 1 & 1
  \end{pmatrix}\right\rangle.
$$  
The code $C_1$ is an MRD code with minimum distance $\alpha_4(C_1)=d_{\min}(C_1)=4$. 
Let $C_2\subseteq\FF_2^{4\times 4}$ be a Gabidulin code of the same dimension and minimum distance as $C_1$.

In~\cite[Example 7.2]{BR17}, it is shown that
$\alpha_5(C_1)=2$. On the other side, since $C_2$ is a Gabidulin code, it is contained in a Gabidulin code of minimum distance $3$ and dimension $8$. This implies that $\alpha_i(C_2)=3$ for $5\leq i\leq 8$.
\end{example}

Thanks to Example~\ref{ex:unequivMRD}, the analog of Theorem~\ref{thm:notinvariant} holds for rank-metric codes and $q$-matroids or $q$-polymatroids.

\begin{theorem}\label{thm:polymatroid}
Let $V=\mathbb{F}_q^{m\times n}$ and $d$ be the rank metric. In this setting, the sequence of code distances is not an invariant of the $q$-matroid or the $q$-polymatroid associated with a code.
\end{theorem}

As in the case of linear block codes, one can find rank-metric codes whose sequence of code distances contains a constant subsequence of arbitrary length.

\begin{example}\label{25.03.03}
A $[mk,k,m]$ rank-metric Hadamard code $\mathcal{H}_{m,k,q}$ is a linear code with a generator matrix such that the columns of it form a basis of $\mathbb{F}_{q^m}^k$ over $\mathbb{F}_q$. They are constant-weight codes with respect to the rank-metric~\cite[Corollary 2]{Ran20}. In this case, $\alpha_i(\mathcal{H}_{m,k,q})=m$ for any $1\leq i\leq k$. 
\end{example}

The next example shows that the sequence of code distances of a code do not determine the sequence of code distances of its dual code.

\begin{example}\rm
Let $\FF_4=F_2[\eta]$, where $\eta^2=\eta+1$, and consider the subcodes of $\mathbb{F}_4^4$
$$C_1=\langle (1,\eta,0,0),\ (0,1,\eta^2,0)\rangle,$$
$$C_2=\langle (1,\eta,0,0),\ (0,0,1,\eta)\rangle.$$
We have $\alpha_i(C_j)=2$ for any $1\leq i\leq 4$, $1\leq j\leq 2$. The dual codes are
    $$C_1^\perp=\langle (1,\eta^2,1,0), (0,0,0,1)\rangle,$$
    $$C_2^\perp=\langle (1,\eta^2,0,0),(0,0,1,\eta^2)\rangle$$
and one can check that $\alpha_i(C_1^\perp)=1\neq 2=\alpha_i(C_2^\perp)$, $1\leq i\leq 4$.
\end{example}

\section{Greedy code distances and partial distances of polar codes}\label{greedy}

One can define a greedy variant of the code distances as follows.

\begin{definition}
Let $C\subseteq V$ be a $k$-dimensional code. For $1\leq i\leq k$, we define an $\alpha_1$-{\bf greedy subcode} as a subcode generated by a codeword of $C$ of maximum weight and an $\alpha_i$-{\bf greedy subcode} as an $i$-dimensional subcode of $C$ of maximum minimum distance among those which contain an $\alpha_{i-1}$-greedy subcode. 
For $k+1\leq i\leq N$, we define an $\alpha_{k+1}$-{\bf greedy supercode} as a $(k+1)$-dimensional code of maximum minimum distance among those which contain $C$ and an $\alpha_i$-{\bf greedy supercode} as $i$-dimensional supercode of $C$ of maximum minimum distance among those which contain an $\alpha_{i-1}$-greedy supercode. We refer in general to $\alpha_i$-{\bf greedy code} for $1\leq i\leq N$.

For $1\leq i\leq k$, define the $i$-th {\bf greedy subcode distance} of $C$ as
$$\alpha^g_i(C)=\max\{d_{\min}(D)\mid D\subseteq C, D\mbox{ is an $\alpha_i$-greedy subcode of } C\}.$$
For $k+1\leq i\leq N$ define the $i$-th {\bf greedy supercode distance} of $C$ as
$$\alpha^g_i(C)=\max\{d_{\min}(D)\mid D\supseteq C, D\mbox{ is an $\alpha_i$-greedy supercode of } C\}.$$
We refer to greedy subcode distances and greedy supercode distances as {\bf greedy code distances}.
\end{definition}

Notice that, by definition, for every $C$ there exists a chain of $\alpha_i$-greedy codes $$D_1\subseteq D_2\subseteq\ldots\subseteq D_{k-1}\subseteq D_k=C\subseteq D_{k+1}\subseteq\ldots\subseteq D_N=V$$ such that $$\alpha_i^g(C)=d_{\min}(D_i)$$ for $1\leq i\leq N$.

It is clear from the definition that 
$$\alpha^g_i(C)\leq\alpha_i(C) \;\;
\mbox{ for } 1\leq i\leq N$$  and 
$$\alpha^g_i(C)=\alpha_i(C) \;\;\; 
\mbox{for $i=1,k,k+1,N$}.$$ 

Similarly to greedy weights and generalized weights, code distances do not agree with their greedy counterparts, in general.

\begin{example}
For $C=V=\FF_2^3$ we have $\alpha_1(C)=3$, $\alpha_2(C)=2$, and $\alpha_1(C)=1$. Moreover, $\alpha^g_1(C)=3$ and $\alpha^g_2(C)=\alpha^g_1(C)=1$.
\end{example}

The difference between code distances and their greedy counterparts boils down to the fact that subcodes and supercodes of $C$ of largest minimum distance do not need to be nested, in general. The next example illustrates this phenomenon.
\begin{example}
If $q$ is odd and $n\leq q$, then $$\alpha_i(\mathbb{F}_q^n)=\alpha_i^g(\mathbb{F}_q^n)=n-i+1.$$ However, for $n=q+1$, $$\alpha_i(\mathbb{F}_q^{q+1})=q-i+2>\alpha_i^g(\mathbb{F}_q^{q+1})=q-i+1\;\; \mbox{ for } 1<i<q+1,$$ while $$\alpha_1(\mathbb{F}_q^{q+1})=\alpha_1^g(\mathbb{F}_q^{q+1})=q+1\; \mbox{ and }\;\alpha_{q+1}(\mathbb{F}_q^{q+1})=\alpha_{q+1}^g(\mathbb{F}_q^{q+1})=1.$$
This depends on the fact that, for $n\leq q$, $\mathbb{F}_q^n$ contains a maximal chain of nested MDS codes, namely the Reed-Solomon codes. For $n=q+1$ and $1<i<(\sqrt{q}+4)/9$, instead, the only $i$-dimensional MDS codes are generalized extended Reed-Solomon codes~\cite{SR03}, which do not contain a codeword of weight $q+1$. This proves that $\alpha_i^g(\mathbb{F}_q^{q+1})\leq q-i+1$ for $1<i<q+1$. The result follows by taking the nested codes $C_i=\langle (\alpha_1^j,\ldots,\alpha_q^j,\delta_{0j})\ :\ 0\leq j\leq i-1\rangle$, $1\leq i\leq q$, where $\delta_{0j}$ is the Kronecker $\delta$, and $\mathbb{F}_q=\{\alpha_1,\ldots,\alpha_q\}$.
\end{example}

\begin{remark}
It follows directly from the definition that if $D_i$ is an $\alpha_i$-greedy subcode of $C$, then $\alpha_j^g(D_i)=\alpha_j^g(C)$ for $1\leq j\leq i$. 
\end{remark}

One can give a simple lower bound on greedy code distances, hence also on code distances. 

\begin{proposition}
Let $C\subseteq V$ be a code and let $c_1,\ldots,c_k$ be a basis of $C$. Suppose without loss of generality that $\wt(c_1)\geq\wt(c_2)\geq\ldots\geq\wt(c_k)$. Then $$\alpha_i^g(C)\geq d_{\min}(\langle c_1,\ldots,c_k\rangle).$$
\end{proposition}

We now compare greedy code weights and partial distances. We start by recalling the definition of partial distances.

\begin{definition}
Let $A\in\mathbb{F}_q^{k\times n}$ be a full-rank matrix with rows $a_1,\ldots,a_k\in\FF_q^n$. The {\bf $i$-th partial distance} $\delta_i$ is defined as
$$\delta_i=\min\{wt(v)\ :\ v\in a_i+\langle a_1,\ldots, a_{i-1}\rangle\}.$$
\end{definition}

Partial distances were defined for non-singular matrices in the setting of polar codes~\cite{MT10, KSR10}. Any polar code is a code generated by some rows of a non-singular matrix $A$ (called the kernel) that polarizes a given channel $W$. 

For a matrix $A$ as above, the partial distances measure the probability of block error of a polar code constructed from $A$~\cite{AM14}.

\begin{definition}\rm
Let $A\in GL_q(n)$ with partial distances $\delta_1,\ldots, \delta_n$. The {\bf exponent} of $A$ is 
$$E(A)=\frac{1}{n}\sum_{i=1}^n \log_n(\delta_i).$$
\end{definition}

The next result follows from the Singleton Bound.

\begin{lemma}
The maximum exponent $E_{n,q}$ for an invertible matrix of size $n$ over a field $\FF_q$ is
$$E_{n,q}\leq \frac{1}{n}\sum_{i=1}^n \log_n(n-i+1).$$
\end{lemma}

The upper bound can be met for given $n$ and $q$ if and only if there is a flag of MDS codes $0\subsetneq D_1\subsetneq\ldots\subsetneq D_{n-1}\subsetneq\mathbb{F}_q^n$. In particular, the maximum $E_{n,q}$ can be attained for $1\leq n\leq q$, but not for $n=q+1$. For fixed values of $n$ which are larger than $q$, it has been a challenge to find matrices that attain the maximum exponent $E_{n,q}$, see~\cite{LLK15}.

The concept of greedy subcode distances is related to that of partial distances of a polar code as follows. 

\begin{theorem}\label{thm:polar}
Let $C\subseteq\FF_q^n$ be a linear block code. There exists a generator matrix $G$ of $C$ such that the corresponding partial distances coincide with the greedy subcode distances of $C$. In addition, $\alpha_i^g(C)$ equals the weight of the $i$-th row of $G$.
\end{theorem}

\begin{proof}
Let $$D_1=\langle g_1\rangle\subseteq D_2=\langle g_{1},g_2\rangle\subseteq\ldots\subseteq D_k=C=\langle g_1,\ldots,g_k\rangle$$ be a chain of $\alpha_i$-greedy subcodes of $C$, i.e., $d_{\min}(D_i)=\alpha_i^g(C)$ for $1\leq i\leq k$.

We claim that $g_1,\ldots,g_k$ can be chosen so that $d_{\min}(D_i)=wt(g_i)$. We proceed by induction on $i$. For $i=1$, any generator $g_1$ of $D_1$ has $\wt(g_1)=d_{\min}(D_1)$. Suppose then that the claim holds true up to a given $i$. If $D_i$ is a constant-weight code, the result is trivial. Hence we suppose that $D_i$ is not a constant-weight code. If $d_{\min}(D_i)>d_{\min}(D_{i+1})$, then a codeword of minimum weight in $D_{i+1}$ is a linear combination of $\{g_1,\ldots,g_{i+1}\}$, with the coefficient of $g_{i+1}$ being nonzero. Hence we can replace $g_{i+1}$ with such minimum weight codeword. If instead $d_{\min}(D_i)=d_{\min}(D_{i+1})$, 
interchange the roles of $g_i$ and $g_{i+1}$ to obtain a chain 
$$D_1\subseteq\ldots\subseteq D_{i-1}\subseteq D_i^\prime=\langle g_1,\ldots,g_{i-1},g_{i+1}\rangle\subseteq D_{i+1}\subseteq\ldots\subseteq D_k=C.$$
Since $d_{\min}(D_{i})=d_{\min}(D_{i+1})\leq d_{\min}(D_i^\prime)\leq\alpha_i^g(C)$, then $d_{\min}(D_{i})=d_{\min}(D_{i}^\prime)$. If $d_{\min}(D_i)<d_{\min}(D_{i-1})$, then, as above, we can add to $g_{i+1}$ a suitable linear combination of $g_1,\ldots,g_{i-1}$ to obtain a vector of weight $d_{\min}(D_i)$. Else, we exchange the roles of $g_{i+1}$ and $g_{i-1}$ and proceed as before. The process terminates at the first $j<i$ such that $d_{\min}(D_j)<d_{\min}(D_{j-1})$. Such a $j$ exists, since, by assumption, $D_i$ is not a constant-weight code. This establishes the claim.

Let $G$ be the matrix with rows $\{g_1,\ldots,g_k\}$ such that $d_{\min}(D_i)=wt(g_i)$ and let. Since the $i$-th partial distance of $G$ is lower bounded by the minimum distance of $D_i$, but $g_{k-i+1}$ has exactly that weight, then these two numbers are equal and we have finished the proof.
\end{proof}

The next corollary shows that greedy code distances and code distances behave differently with respect to containment, see also Remark~\ref{rmk:subcode}.

\begin{corollary}
Let $C\subseteq\FF_q^n$ be a linear block code, fix $1\leq j\leq k$. Then $C$ has a subcode whose greedy subcode distances are the first $j$ greedy subcode distances of $C$.
\end{corollary}

\begin{proof}
Let $D_1\subseteq\ldots\subseteq D_i\subseteq D_k=C$ be a chain of $\alpha_i$-greedy subcodes of $C$ and consider $D_j\subseteq C$.
The result follows from the previous theorem, after noticing that there is a one-to one correspondence between chains $D_1^\prime\subseteq\ldots\subseteq D_{j-1}^\prime \subseteq D_j$  of subcodes of $D_j$ and chains of subcodes of $C$ of the form $D_1^\prime\subseteq\ldots\subseteq D_{j-1}^\prime \subseteq D_j\subseteq D_{j+1}\subseteq\ldots\subseteq D_k=C$.
\end{proof}

The greedy weights of $\mathbb{F}_q^n$ provide a constructive lower bound on the maximum exponent of an invertible matrix of size $n$. More precisely, we have the

\begin{proposition}
    Let $n$ be a positive integer and let $E=\max\{E(M)\ :\ M\in GL_q(n)\}$. Then 
    $$\frac{1}{n}\sum_{i=1}^n\log_n(\alpha_i^g(\mathbb{F}_q^n))\leq E$$
\end{proposition}

\begin{proof}
By the previous theorem, there exists an $A$ such that the partial distances of $A$ are $\alpha_i^g(\mathbb{F}_q^n)$, thus
    $$E(A)=\frac{1}{n}\sum_{i=1}^n\log_n(\alpha_i^g(\mathbb{F}_q^n))\leq E.$$
\end{proof}


\begin{example}
In $\mathbb{F}_2^3$, the matrix with the largest exponent has partial distances $2,2,1$. The sequence of greedy code distances of $\FF_2^3$ is $3,1,1$ and its sequence of code distances is $3,2,1$.
\end{example}

\section{Comparison with related code invariants}\label{comparison}

In this section, we compare code distances with some related code invariants. We start by recalling the definition of maximal code.

\begin{definition}
A code $C\subseteq V$ is {\bf maximal} if either $C=0$, or there exists no $D\supsetneq C$ such that $d_{\min}(D)=d_{\min}(C)$. 
\end{definition}

The definition of maximal code can be easily rephrased in terms of code distances. More in general, one can translate a strict inequality between consecutive greedy code distances into a statement on the maximality of certain subcodes or supercodes of~$C$.

\begin{proposition}
Let $0\neq C\subseteq V$ be a $k$-dimensional code, let $1\leq i<k$ and $j\geq k$. Then:
\begin{enumerate}
\item $C$ is maximal if and only if $\alpha_{k+1}(C)<\alpha_k(C)$,
\item if every $\alpha_i$-greedy subcode of $C$ is maximal, then $\alpha_{i+1}(C)<\alpha_i(C)$,
\item every $\alpha_j$-greedy supercode of $C$ is maximal if and only if $\alpha_{j+1}(C)<\alpha_j(C)$.
\end{enumerate}
\end{proposition}

The next example shows that the converse of the implication in part (2) does not hold in general.

\begin{example}
Let $C=\langle(1,0,0),(0,1,0)\rangle\subseteq\FF_2^3$. One can check that $\alpha_1^g(C)=2$ and $\alpha_2^g(C)=\alpha_3^g(C)=1$. The only $\alpha_1$-greedy subcode of $C$ is $D_1=\langle(1,1,0)\rangle$, which is not maximal, since it is contained in the even weight code in $\FF_2^3$. 

Notice moreover that, since $\alpha_2^g(C)=\alpha_3^g(C)=1$, we expect $C$ to have a non-maximal $\alpha_2$-greedy supercode. This is indeed the case, as the only $\alpha_2$-greedy supercode of $C$ is $C$ itself, which is not maximal.
\end{example}

In the case of rank-metric codes, an invariant related to the concept of maximality was defined by Byrne and Ravagnani in~\cite[Definition~2.2]{BR17}. We extend their definition to our more general setting.

\begin{definition}
Let $0\neq C \subsetneq V$ be a code. The {\bf maximality degree} of $C$ is 
$$\mu(C)=\min\{d_{\min}(C)-d_{\min}(D)\mid D\supsetneq C \mbox{ is a code}\}.$$
\end{definition}

By definition, $0\leq \mu(C)\leq d_{\min}(C)-1$. Moreover, $\mu(C)>0$ if and only if $C$ is a maximal code.
It is easy to relate the maximality degree with the first supercode distance.

\begin{proposition}\label{prop:mu}
Let $0\neq C\subseteq V$ be a $k$-dimensional code.
Then $$\alpha_{k+1}(C)=d_{\min}(C)-\mu(C).$$
\end{proposition}

\begin{proof}
It follows directly from the definition that $d_{\min}(C)-\alpha_{k+1}(C)\geq\mu(C)$, that is $\alpha_{k+1}(C)\leq d_{\min}(C)-\mu(C)$. Moreover, let $D\supsetneq C$ be a code such that $\mu(C)=d_{\min}(C)-d_{\min}(D)$ and let $h=\dim(D)$. Then $h\geq k+1$ and $d_{\min}(C)-\mu(C)=d_{\min}(D)=\alpha_h(C)\leq\alpha_{k+1}(C)$.
\end{proof}

The covering radius is an invariant of a code which is not an invariant of its associated matroid for codes endowed with the Hamming metric and not an invariant of the associated $q$-matroid or $q$-polymatroid for rank-metric codes. This was observed in~\cite{BR05} and~\cite[Example~6.8]{GJLR20}, respectively.

\begin{definition}
Let $0\neq C\subseteq V$ be a code. The {\bf covering radius} of $C$ is $$\rho(C)=\max\{d(C,v)\mid v\in V\}.$$
\end{definition}

We now propose a family of invariants that extend the covering radius. Notice that these are different from generalized covering radii, as defined in~\cite{EFS21}. The relation between code distances and covering radii will be discussed in the last section.

\begin{definition}
Let $C\subseteq V$ be a $k$-dimensional code. For $1\leq i\leq N-k$, we define the $i$-th {\bf generalized radius} of $C$ as
$$\rho_i(C)=\max\{d(C,D)\mid\dim D=i\},$$ where $d(C,D)=\min\{d(x,y)\mid x\in C, y\in D\setminus 0\}$.
\end{definition}

\begin{remark}
Notice that $$\rho_i(C)=\max\{d(C,D)\mid D\cap C=0,\ \dim D=i\}.$$ In fact, if $D\cap C\neq 0$, then $d(C,D)=0$. Notice moreover that, if $C\cap D=0$, then $$d(C,D)=\min\{\wt(v)\mid v\in (C+D)\setminus C\}.$$ This shows in particular that $d(C,D)$ and $d(D,C)$ are not equal, in general.
\end{remark}

The next proposition follows directly from the definition.

\begin{proposition}
Let $C\subseteq V$ be a $k$-dimensional code. Then,
$$\rho(C)=\rho_1(C)\geq\rho_2(C)\geq\ldots\geq\rho_{N-k}(C)=\min\{\wt(v)\mid v\not\in C\}.$$
In particular, if $V$ is generated by elements of weight one, then $\rho_{N-k}(C)=1$.
\end{proposition}

We establish a relation between generalized radii and (greedy) generalized subcode distances. 

\begin{theorem}\label{thm:radii}
For any $k$-dimensional code $C\subseteq V$ and $1\leq i\leq N-k$, we have
$$\alpha_{k+i}(C)=\min\{d_{\min}(C),\rho_i(C)\}$$ and $$\alpha_{k+i}^g(C)=\min\{\alpha^g_{k+i-1}(C),\rho(D_{k+i-1})\},$$ where $D_1\subseteq D_2\subseteq\ldots\subseteq D_k=C$ is a chain of $\alpha_i$-greedy subcodes of $C$.
\end{theorem}

\begin{proof}
Let $D\subseteq\FF_q^n$ be an $i$-dimensional code such that $C\cap D=0$. Then $C+D$ is a $(k+i)$-dimensional code with $$d_{\min}(C+D)=\min\{d_{\min}(C),d(C,D)\}.$$ In fact, for any $v\in C+D\setminus 0$ one has $\wt(v)\geq d_{\min}(C)$ if $v\in C\setminus 0$ and $\wt(v)\geq d(C,D)$ if $v\in C+D\setminus C$. Moreover, there exist $u\in C\subseteq C+D$ with $\wt(u)=d_{\min}(C)$ and $w\in C+D\setminus C$ with $\wt(w)=d(C,D)$. 

By definition of generalized radius, for any $i$-dimensional code $D\subseteq V$ such that $C\cap D=0$ one has $d(C,D)\leq\rho_i(C)$, hence the largest value of $d_{\min}(C+D)$ is achieved by a $D$ with $d(C,D)=\rho_i(C)$ and is $$d_{\min}(C+D)=\min\{d_{\min}(C),\rho_i(C)\}.$$ This proves the first equality.

The second equality follows from the first one for $i=1$, since $\alpha_{k+1}^g(C)=\alpha_{k+1}(C)$, $\alpha_{k}^g(C)=\alpha_{k}(C)=d_{\min}(C)$, and $\rho_1(C)=\rho(D_k)$ as $D_k=C$. In general, $$\alpha_{k+i}^g(C)=\alpha_{k+i}(D_{k+i-1})=\min\{d_{\min}(D_{k+i-1}),\rho(D_{k+i-1})\}
=\min\{\alpha^g_{k+i-1}(C),\rho(D_{k+i-1})\},$$
where the first and third equalities follow from the definition of greedy code distances and the second equality follows from the first part of the statement.   
\end{proof}

For a block code endowed with the Hamming metric, we now define its strongly linearly dependent set. 

\begin{definition}
Let $H$ be a parity check matrix of $C \subseteq \FF_q^n$. The {\bf strongly linearly dependent} set of $C$ is
\[SLD(C) = \{ s \in \mathbb{N} : s \text{ columns of } H \text{ are } s.l.d. \}.\]
\end{definition}

Notice that, even though the strongly linear dependent set is defined in terms of a parity check matrix, it does not depend on it, since it is the set of cardinalities of supports of the nonzero codewords.

\begin{example}\label{25.03.07}
Let $C\subseteq \FF_2^4$ be the even weight code, whose parity check matrix is $H=(1\; 1\; 1\; 1)$. Any two or four columns of $H$ are s.l.d., and any single or three columns of $H$ are not s.l.d. Then,
\[SLD(C) = \{2,4\}.\]
\end{example}

\begin{example}\label{25.03.08}
Let $C = \FF_2^3$. The parity check matrix is $H=(0\; 0\; 0)$. As any number of columns is s.l.d., then
\[SLD(C) = \{1,2,3\}.\]
\end{example}

The next result shows that the code distances are elements of $SLD(C)$. In particular, $\alpha_1(C)$ and $\alpha_k(C)$ can be obtained from the strongly linear dependent set of~$C$.

\begin{proposition}\label{alpha strong}
Let $C \subseteq V$ be a linear code and let $H$ be its parity check matrix. Then
\[\{\alpha_1(C), \ldots, \alpha_k(C)\} \subseteq SLD(C).\]
Moreover, $\alpha_1(C)$ and $\alpha_k(C)$ are respectively the maximum and minimum of $SLD(C)$.
\end{proposition}

\begin{proof}
For every $\alpha_i(C)$, $1\leq i\leq k$, there is an element $c=(c_1\ldots,c_n) \in C$ such that $\alpha_i(C) = \wt(c)$. Let $H$ be a parity check matrix of $C$. As the columns of $H$ in position $j$ such that $c_j \neq 0$ are s.l.d., then $\alpha_i(C) \in SLD(C)$. In particular, $\alpha_1(C)$ and $\alpha_k(C)$ are respectively the maximum and minimum number of strongly linearly dependent columns of $H$ by Corollary~\ref{25.03.01} and Proposition~\ref{alphainequality}.
\end{proof}

\section{Asymptotic code distances}\label{asymtotic}

In this section, for a code defined over $\FF_q$, we define code distances for the extension of the code to every finite field that contains $\FF_q$. 

\begin{definition}
Let $C\subseteq V$ be a $k$-dimensional code. For $\ell\geq 1$, let $C_\ell=C\otimes_{\FF_q}\FF_{q^\ell}\subseteq\FF_{q^\ell}^n$ be the {\bf extension} of $C$ to $\FF_{q^\ell}$.
Let $1\leq i\leq n$. The $i$-th {\bf extended code distance} of $C$ relative to the field extension $\FF_{q^\ell}/\FF_q$ is 
$$\alpha_i^\ell(C)=\alpha_i(C_{\ell}).$$
The $i$-th {\bf asymptotic code distance} of $C$ is 
$$\alpha^\infty_i(C)=\max\{\alpha_i^\ell(C)\ : \ \ell\geq 1\}.$$
\end{definition}

Notice that for the Hamming, rank, and sum-rank distances $$\alpha_k(C)=d_{\min}(C)=d_{\min}(C_\ell)=\alpha_k^\ell(C)\; \mbox{ for all } \ell\geq 1.$$
However, extended code distances may differ from the corresponding code distances, as the next example shows.

\begin{example}
Let $C=\langle(1,1,0),(0,1,1)\rangle\subseteq\FF_2^3$ be the even weight code. Then $$\alpha_1(C)=2<3=\alpha_1^2(C).$$ In fact, $C$ contains no word of weight $3$, but $C\otimes_{\FF_2}\FF_2(\eta)\ni (1,1+\eta,\eta)$.
\end{example}

Notice moreover that asymptotic code distances exist and are finite whenever the weight function on $V$ only takes a finite number of values. This is the case, e.g., when $V=\FF_q^n$ and $d$ is the Hamming distance, when $V=\FF_q^{m\times n}$ and $d$ is the rank distance, or when $V=\FF_q^{m_1\times n_1}\times\cdots\times\FF_q^{m_\ell\times n_\ell}$ and $d$ is the sum-rank distance.

\medskip 

Code distances and extended code distances are related as follows under a mild hypothesis. The hypothesis is satisfied by many metrics, including the Hamming, rank, and sum-rank distances.

\begin{theorem}
Let $C\subseteq V$ be a code and assume that the weight function is such that $d_{\min}(D)=d_{\min}(D\otimes\mathbb{F}_{q^\ell})$ for any $D\subseteq V$. Then 
$$\alpha_i(C)\leq \alpha_i^{\ell}(C)\leq \alpha^\infty_i(C)$$
for $1\leq i\leq N$ and $\ell\geq 1$.  
Moreover, for every $1\leq i\leq N$ and every $\ell\geq 1$ there exists $\lambda\geq \ell$ such that
$$\alpha^\infty_i(C)=\alpha_i^{\lambda}(C).$$
\end{theorem}

\begin{proof}
Let $D\subseteq V$ be an $i$-dimensional code with $d_{\min}(D)=\alpha_i(C)$ and such that $D\subseteq C$ if $i\leq k$ or $D\supseteq C$ if $i\geq k$.
Then $D_\ell=D\otimes_{\FF_q}\FF_{q^\ell}\subseteq C_\ell$ is an $i$-dimensional sub- or supercode of $C$ such that $d_{\min}(D_\ell)=d_{\min}(D)$. Therefore, $\alpha_i^\ell(C)\geq d_{\min}(D_\ell)=\alpha_i(C)$. Moreover, $\alpha_i^\ell(C)\leq \alpha^\infty_i(C)$ by definition. 

Let $t\geq 1$ be such that $\alpha^\infty_i(C)=\alpha_i^t(C)$. Since $\alpha_i^t(C)\leq\alpha_i^{st}(C)\leq\alpha_i^\infty(C)$, one has that $\alpha_i^{st}(C)=\alpha_i^\infty(C)$ for every $s\geq 1$. In particular, one can choose $s$ so that $st\geq\ell$ and let $\lambda=st$.
\end{proof}

For the Hamming case, we can prove a stronger result, i.e., that the sequence of $\{\alpha_i^\ell(C)\}_{\ell=1}^\infty$ stabilizes. 

\begin{theorem}
 Let $C\subseteq\mathbb{F}_q^n$. For any $1\leq i\leq k$, there is $t\geq 1$ such that for any $\ell\geq t$, 
 $$\alpha_i^\ell(C)=\alpha_i^\infty(C).$$
\end{theorem}

\begin{proof}
Let $1\leq i\leq k$ and $j=k-i$ and assume $\alpha_i^\infty(C)=d$. Let $H$ be a parity check matrix of $C$ and let $$\mathbf{x}=\begin{pmatrix} x_{11}&\cdots&x_{1n}\\ \vdots&\ddots&\vdots\\
    x_{j1}&\cdots&x_{jn}\end{pmatrix}.$$

Let $T(\mathbf{x})=\begin{pmatrix} H\\ \mathbf{x}\end{pmatrix}$ and define the following sets of ideals in $\bar{\mathbb{F}_q}[\mathbf{x}]$:
\begin{itemize}
\item $I_{dim}$ the ideal of minors of size $n-i$ of $T(\mathbf{x})$.
\item For each $J\subseteq [n]$ of size $d-1$, let $I_J$ be the ideal of minors of size $d-1$ of the submatrix of $T(\mathbf{x})$ whose columns are indexed by $J$.
\item $$I_{C,i}=I_{\dim}\cap\bigcap_{J\subseteq[n], |J|=d-1} I_J.$$
\end{itemize}

Let $P\in\mathbb{F}_q^{jn}$. The previous ideals provide us with the following information:
\begin{itemize}
\item If $P\notin Z(I_{\dim})$, then $T(P)$ has rank $n-i$ and it defines a parity check for a subcode of $C$ of dimension $i$.
\item If for any $J\subseteq [n]$ of size $d-1$, $P\notin Z(I_J)$, then any $d-1$ columns of $T(P)$ are linearly independent, hence the code with parity check matrix $T(P)$ has minimum distance at least $d$.
\item By the previous two items, if $P\notin Z(I_{C,i})$, then $T(P)$ is the parity check matrix of a subcode of $C$ of dimension $i$ and minimum distance greater than or equal to $d$. Equality follows from the maximality of $\alpha_i^\infty(C)$.
\end{itemize}

Let $\mu(I_{C,i})$ be the minimum degree of an element of $I_{C,i}$ and let $$t=\min\{r\ :\ \mu(I_{C,i})< q^r\}.$$

Then for any $\ell\geq t$ one has $I_{C,i}\subseteq\sqrt {I_{C,i}}\not\subseteq (x_{gh}^{q^\ell}-x_{gh})_{g,h=1}^{j,n}$, hence $\mathbb{F}_{q^\ell}^n\setminus Z(I_{C,i})\neq \emptyset$. This guarantees that for any $\ell\geq t$, $\alpha_i^\ell(C)=d$.
\end{proof}

\begin{example}
Let $C_1$ and $C_2$ be the binary codes of Example~\ref{ex:duality}. We already remarked in Example~\ref{ex:duality} that $C_1$ and $C_2$ have the same code distances. However, they do not have the same asymptotic code distances, as $$\alpha_1^\infty(C_1)=\alpha_1^t(C_1)=5$$ while $$\alpha_1^\infty(C_2)=\alpha_1^t(C_2)=4$$ for any $t\geq 2$. 
Moreover, $2=\alpha_j(C_i)\leq \alpha_j^t(C_i)\leq d_{\min}(C_i)=2$ for every $j\neq 1,5$, $t\geq 1$, and $i=1,2$, showing that $\alpha_j^\infty(C_i)=2$ for $i=1,2$. Finally, $\alpha_5^\infty(C_i)=1$ for $i=1,2$.
\end{example}

\begin{theorem}\label{thm:chainMDS}
Let $V=\FF_q^n$ and let $d$ be the Hamming distance. Let $C\subseteq V$ be a $k$-dimensional MDS code. Then there exists a $t\geq 1$ such that $\alpha_{i}^t(C)=n-i+1$ for $1\leq i\leq k$.
\end{theorem}

\begin{proof}
We prove the statement by induction on $k\geq 1$. If $k=1$, then the statement trivially holds. Assume by induction that the statement holds for $k-1$ and let $H$ be a parity check matrix of $C$. The matrix $H^\prime$ obtained from $H$ by deleting the last column is the parity check matrix of a puncturing of $C$, say $D\subseteq\FF_q^{n-1}$. By the induction hypothesis, there exists an $s$ such that $D\otimes_{\FF_q}\FF_q^s$ has an MDS subcode of codimension one. In other words, there is an element $(h_1,\ldots,h_{n-1})\in\FF_{q^s}^{n-1}$ that one can add as a row to $H^\prime$ and such that the resulting matrix is the parity check matrix of an MDS code $D^\prime$. In order to produce an MDS subcode of $C\otimes_{\FF_q}\FF_q^{t}$, one needs to find $t$ such that $s\mid t$ and an element $x\in\FF_q^{t}$ with the property that every maximal minor of the matrix obtained by appending the row $(h_1,\ldots,h_{n-1},x)$ to $H$ is nonzero. Every minor which does not involve the last column is nonzero, since $D^\prime$ is MDS. Every minor which involves the last column yields a linear equation in $x$ where the coefficient of $x$ is nonzero, since it is a maximal minor of $H^\prime$. Therefore, over a large enough field extension $\FF_{q^t}$ of $\FF_{q^s}$ one finds $x\in\FF_q^{t}$ which makes all minors nonzero. This proves that $\alpha_{k-1}^{t}(C)=n-k+2$. Repeatedly applying this result yields the thesis.
\end{proof}

\begin{remark}
Notice that the same argument used in the proof of Theorem \ref{thm:chainMDS} shows that, for $q\gg 0$, any MDS code in $\FF_q^n$ fits in a maximal chain of MDS codes. In particular, if one takes $q$ larger than the number of maximal minors of the matrix obtained by appending the row $(h_1,\ldots,h_{n-1},x)$ to $H$ which involve the last column, then there is a value of $x$ which makes all these minors nonzero. Notice that the number of such minors only depends on $k,n$. This shows that, for $q\gg 0$, any MDS code in $\FF_q^n$ contains a maximal chain of MDS subcodes. By applying the same result to the dual code, one sees that the original code fits in a maximal chain of MDS codes. 
\end{remark}

Generalized covering radii were defined in~\cite[Definition 8]{EFS21} for codes endowed with the Hamming metric. We extend the definition in the obvious way to the larger family of codes that we consider and to field extensions of arbitrary degrees.

\begin{definition}\label{def:gencovrad}
Let $C\subseteq V$ be a code, let $i\geq 1$. The $i$-th {\bf generalized covering radius} of $C$ is 
$$R_\ell(C)=\rho(C_\ell).$$
\end{definition}

This is related to our concept of asymptotic code distances as follows.

\begin{theorem}\label{thm:genradii}
Let $C\subseteq V$ be a $k$-dimensional code. Then 
$$\alpha_{k+1}^\ell(C)=\min\{d_{\min}(C),R_\ell(C)\}$$
for any $\ell\geq 1$.
\end{theorem}

\begin{proof}
We have the chain of equalities
$$\alpha_{k+1}^\ell(C)=\alpha_{k+1}(C_\ell)=\min\{d_{\min}(C_\ell),\rho(C_\ell)\}=\min\{d_{\min}(C),R_\ell(C)\}$$
where the second equality follows from Theorem~\ref{thm:radii} and the others follows from the respective definitions.
\end{proof}

{\bf Acknowledgements:} The authors thank Andrea Di Giusto for useful discussions on the material of this paper.


\end{document}